\documentclass[superscriptaddress,twocolumn]{revtex4}
\usepackage{graphicx}
\usepackage{svg}
\usepackage{color}\usepackage{enumerate}
\usepackage{amsmath}\usepackage{amssymb}\usepackage{stmaryrd}
\usepackage{amsthm}
\usepackage{esint}
\usepackage{multirow}
\usepackage{float}
\usepackage{booktabs}
\usepackage{dsfont}
\usepackage{amssymb}
\usepackage{comment}
\usepackage{parskip}
\usepackage{physics}
\usepackage[pdftex,colorlinks=true,linkcolor=blue,citecolor=blue]{hyperref}
\usepackage{cleveref}
\usepackage{bbold}
\usepackage{mathtools}
\usepackage{amsmath}
\usepackage{caption,subcaption,graphicx}
\newtheorem{theorem}{Theorem}[section]

\newtheorem{lemma}[theorem]{Lemma}

\PassOptionsToPackage{hyphens}{url}\usepackage{hyperref}

\begin{document}
\title{Learning dynamic quantum circuits for efficient state preparation}

\author{Faisal Alam}
\affiliation{Department of Physics, University of Illinois at Urbana-Champaign, IL 61801, USA}
\affiliation{IQUIST and Institute for Condensed Matter Theory and NCSA Center for Artificial Intelligence Innovation,
University of Illinois at Urbana-Champaign, IL 61801, USA}
\affiliation{Theoretical Division, Los Alamos National Laboratory, Los Alamos, NM 87545, USA}

\author{Bryan K. Clark}
\affiliation{Department of Physics, University of Illinois at Urbana-Champaign, IL 61801, USA}
\affiliation{IQUIST and Institute for Condensed Matter Theory and NCSA Center for Artificial Intelligence Innovation,
University of Illinois at Urbana-Champaign, IL 61801, USA}

\begin{abstract}
Dynamic quantum circuits (DQCs) incorporate mid-circuit measurements and gates conditioned on these measurement outcomes. DQCs can prepare certain long-range entangled states in constant depth, making them a promising route to preparing complex quantum states on devices with a limited coherence time. Almost all constructions of DQCs for state preparation have been formulated analytically, relying on special structure in the target states. In this work, we approach the problem of state preparation variationally, developing scalable tensor network algorithms which find high-fidelity DQC preparations for generic states. We apply our algorithms to critical states, random matrix product states, and subset states. We consider both DQCs with a fixed number of ancillae and those with an extensive number of ancillae. Even in the few ancillae regime, the DQCs discovered by our algorithms consistently prepare states with lower infidelity than a static quantum circuit of the same depth.  Notably, we observe constant fidelity gains across system sizes and circuit depths. For DQCs with an extensive number of ancillae, we introduce scalable methods for decoding measurement outcomes, including a neural network decoder and a real-time decoding protocol. Our work demonstrates the power of an algorithmic approach to generating DQC circuits, broadening their scope of applications to new areas of quantum computing. 
\end{abstract}

\maketitle

\section{Introduction}
The axioms of quantum mechanics allow two means to evolve a quantum state: unitary operators and measurements. A circuit-based quantum computer utilizes local unitary operators, or gates, and only performs measurements at the end of an algorithm. On the other hand, a measurement-based quantum computer \cite{mbqc} begins with a highly entangled resource state and applies only a sequence of measurements to it. A dynamic quantum circuit (DQC) incorporates both paradigms by allowing measurements in the middle of a gate-based circuit, followed by further gates chosen adaptively based on classical processing of the measurement outcome. DQCs are also sometimes referred to as circuits assisted by local operations and classical communication (LOCC) or circuits with measurement and feedforward. 
\\ \\
Originally introduced in the context of the quantum teleportation protocol \cite{nielsen_chuang}, dynamic quantum circuits are a crucial milestone on the roadmap to fault-tolerant quantum computers \cite{qec_old, qec_new} since quantum error correction schemes rely on applying correction gates conditioned on syndrome measurements. Dynamic circuits have recently also found application in algorithms like quantum singular value transform and quantum Fourier transform \cite{noise_limiting, qsvt, constant_exp, qft}. Several quantum platforms have successfully demonstrated the implementation of these circuits \cite{qec_expr1, qec_expr2, qec_expr3}.
\\ \\ 
An exciting area of applications for DQCs is in the preparation of long-range entangled states. The Lieb-Robinson bound \cite{lieb_robinson, qi_qm} tells us that with local gates, the circuit depth required to prepare such states scales with system size. Since near-term quantum computers have limited coherence times, and therefore, can only execute shallow circuits reliably, the system sizes one can prepare are severely restricted. However, recent work has shown that certain long-range entangled states, like the GHZ state, topologically ordered states \cite{aklt_from_locc, phases_from_locc, lre_from_locc, topo_from_locc, spt_locc} and matrix product states with certain symmetries \cite{mps_from_locc, mps_peps_from_locc, tn_from_locc, log_depth_mps}, can be prepared by dynamic quantum circuits with depth independent of system size. This is because classical communication of measurement outcome across the system violates locality and, therefore, makes the Lieb-Robinson bound inapplicable. Proposals for DQC preparation of states have also been implemented experimentally \cite{expr1, expr2}.
\\ \\ 
Prior work has exploited circuit identities, the stabilizer formalism and group theory to construct DQCs. Such approaches are successful only for very structured states and require independent effort for each new class of states. Furthermore, there is no guarantee that the construction found is optimal in the presence of hardware constraints like gate sets, noise and a limited number of qubits.
\\
\begin{figure*}[!t]
  \centering
  \includegraphics[width=\linewidth]{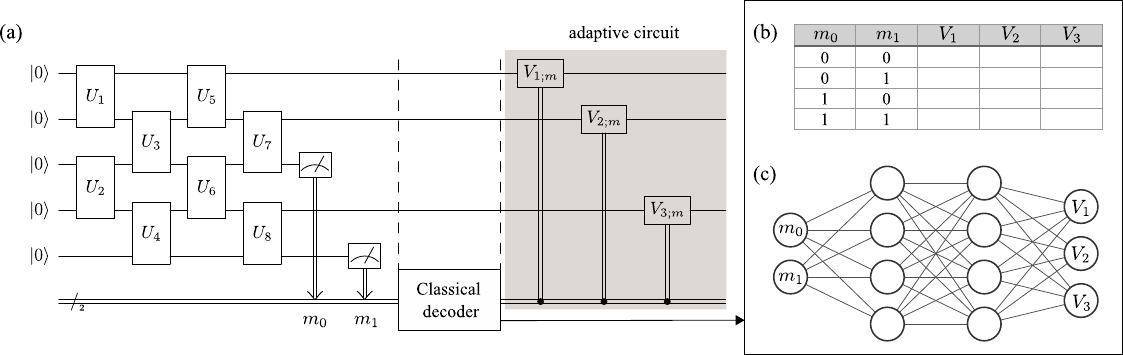}
  \caption{A dynamic quantum circuit ansatz: (a) Variational gates, $U_j$, are applied, followed by measurement on ancilla qubits. The measurement outcomes, $m$, are then passed to a variational decoder which recommends additional gates, $V_{k;m}$. (b) A look-up table decoder tabulates gates with parameters optimized for each measurement outcome. (c) A feedforward neural network decoder, with trained weights and biases, generates appropriate gates based on measurement outcomes.}
  \label{ansatz}
\end{figure*}
\\
In this work we propose algorithmic solutions to the problem of preparing states with dynamic circuits. We train parameterized DQCs to maximize fidelity with a target state. Our algorithms are classical and are scalable due to extensive use of tensor network methods. The algorithms can accept generic target states, although  DQCs only have an advantage over static circuits when the target state has small entanglement and high complexity. Our algorithms can be easily modified to take into account hardware constraints, but we do not fully explore this possibility in the current work. 
\\ \\ 
The first algorithm we present optimizes DQCs that incorporate a few mid-circuit measurements. We apply this algorithm to three classes of target states not previously explored in the dynamic circuit literature: critical states, generic matrix product states and subset states. Our results show that DQCs achieve significantly lower infidelity than static circuits of the same depth. While we do not expect DQCs with few mid-circuit measurements to prepare states in constant depth, the success of our algorithm indicates that limited measurement and feedforward can substantially improve state preparation fidelities.
\\ \\ 
We then introduce two protocols for training DQCs with an extensive number of mid-circuit measurements. These protocols are designed to manage the exponential growth of measurement outcomes with system size. We benchmark their performance using the paradigmatic GHZ state.
\\ \\ 
Variational algorithms with quantum circuits \cite{vqa} have seen numerous applications in the noisy-intermediate scale quantum computing (NISQ) era, including ground state search \cite{vqe_orig, vqe_review}, compilation of unitaries \cite{fumc, fisc}, real and imaginary time evolution \cite{pvqd, tubos}, and machine learning \cite{qnn_new, qnn_old}. We expect our work to lead to further studies of how dynamic quantum circuits can provide new algorithms for such applications. 

\section{Dynamic quantum circuits}

In order to prepare an $n$ qubit quantum state, we use dynamic quantum circuits (see Fig \ref{ansatz}) with $n$ system qubits and some number of ancilla qubits. After a few layers of entangling gates are applied, the ancillae are measured. The measurement outcome is passed to a classical function, called a decoder, which returns a description of gates to apply to the unmeasured qubits, often in the form of rotation angles. 
\\ \\ 
The two trainable components in a DQC ansatz are the parametrized gates prior to measurements and the classical decoder. The number of ancillae used is an important architectural hyperparameter of dynamic circuits and influences the choice of classical decoder. In Section \ref{few_ancillae}, we present a look-up table decoder designed for DQCs with a limited number of ancillae. In Section \ref{extensive_ancillae}, we explore a neural network decoder and a real-time decoding protocol tailored for DQCs with an extensive number of ancillae.   
\\
\\
Note that since dynamic quantum circuits differ from static ones only by the presence of additional local operations and classical communication (LOCC), they can generate at most the same bipartite entanglement as static circuits of the same depth \cite{locc_basics}. As a result, DQCs offer an advantage only for complex quantum states whose preparation is not limited by entanglement generation. Here, we define complexity of a state as the minimum circuit depth required for its preparation. Consider, for instance, the $n$ qubit GHZ state: despite having only has $\mathcal{O}(1)$ entanglement entropy across any bipartition, its complexity scales as $\mathcal{O}(n)$. It is in the presence of such a gap between complexity and entanglement that DQCs excel. Since classical tensor network methods are similarly useful in efficiently representing complex states with low entanglement, we propose using tensor network algorithms to optimize DQCs for state preparation.  

\section{DQC with few ancillae \label{few_ancillae}}

The simplest decoder for a dynamic quantum circuit is a look-up table that lists the post-measurement circuit, or equivalently its parameters, for each measurement outcome. The decoder is then trained by optimizing each post-measurement circuit such that it rotates the corresponding post-measurement state to the target state. However, note that for $r$ ancilla qubits, the number of measurement outcomes is $2^r$. Therefore, both the size of the look-up table and the number of parameters we must optimize over grows exponentially with $r$ and training is feasible only when the number of ancillae remain few. However, we will soon present result which demonstrates that DQCs already offer substantial advantage over static circuits in this regime. 

\subsection{Optimization method}

In this section we present a classical tensor network algorithm to optimize DQCs with few ancillae.   
\\
\begin{figure*}[!t]
  \centering
  \includegraphics[width=\linewidth]{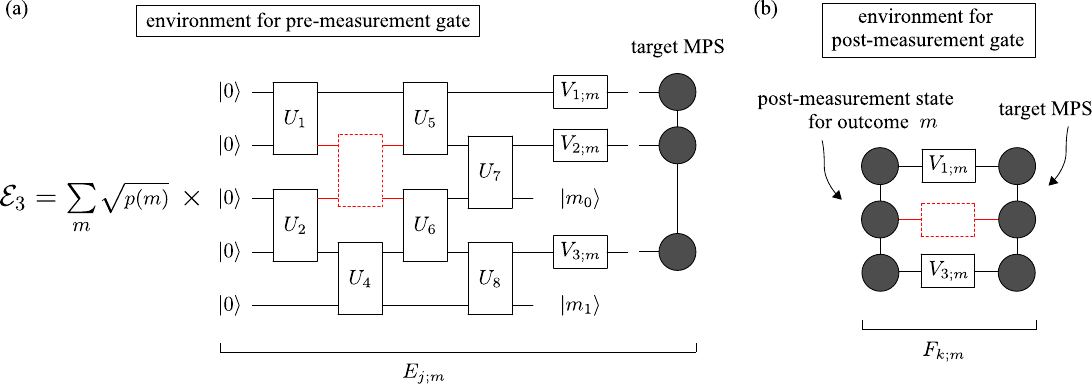}
  \caption{Optimizing dynamic quantum circuits with look-up table decoders for state preparation. Environment tensors are computed and used to update gates via singular value decomposition. (a) Environment tensor for a pre-measurement gate. Projector tensors are applied to ancilla qubits for each measurement outcome, $m$. All tensors, except the chosen gate, are then contracted. The resulting tensor is weighted by the probability of $m$ and summed over all possible outcomes. (b) Environment tensor for an adaptive gate applied post-measurement, corresponding to outcome $m$.}
  \label{environment}
\end{figure*}
\\ 
A standard figure of merit used in state preparation is the fidelity of the prepared state with the target state. Since the output of a randomly initialized DQC is an ensemble of states over the measurement outcomes, the fidelity is defined as $\bra{\psi_{\text{targ}}}\rho_{\text{prep}}\ket{\psi_{\text{targ}}}$, where $\ket{\psi_{\text{targ}}}$ is the target state and $\rho_{\text{prep}}$ is a density matrix representing the DQC ansatz. 
\\ 
\begin{figure}
  \centering
  \includegraphics[width=\linewidth]{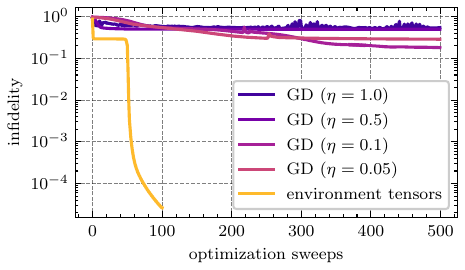}
  \caption{Comparison of gradient descent to the environment tensor method. We optimize a dynamic quantum circuit to learn the GHZ state on 8 qubits with two ancillae. For a range of learning rates, $\eta$, gradient descent either plateaus to a poor optimum or converges slowly, whereas the environment tensor method converges rapidly.}
  \label{gd_vs_svd}
\end{figure}
\\
A natural way to train DQCs for state preparation is to use infidelity as a cost function and perform gradient descent. However, this approach suffers from the (classical) barren plateau problem. Since vectors in large dimensional spaces have small overlaps, gradients of the infidelity for a randomly initialized DQC vanish exponentially with $n$. Our numerical experiments show (see Figure \ref{gd_vs_svd}) that gradient descent often fails to converge to optimal minima in reasonable times. Therefore, while we compute and report infidelity in our results, we instead minimize a related, but distinct quantity defined below, which enables us to use a more effective optimization method. 
\\ \\ 
Our algorithm optimizes DQCs by sweeping over both the pre- and post-measurement gates and performing globally optimal updates for each gate with respect to appropriate cost functions. This process continues until convergence has been reached or a maximum number of iterations completed.  
\\ \\ 
Consider the overlap between the target state and the state output by the DQC for measurement outcome $m$:  
\begin{align}
    o_m &= \bra{\psi_{\text{targ}}} \prod_k V_{k;m}\,\frac{\bra{m}}{\sqrt{p(m)}}\,\prod_j U_j\ket{0}^{\otimes (n+r)} \label{overlap},
\end{align}
where $U_j$ are pre-measurement gates, $V_{k;m}$ are post-measurement gates for outcome $m$ and $\ket{m}$ represents the state of the measured qubits. Probability of the outcome, $p(m)$, is used to ensure normalization. 
\\ \\ 
To train the pre-measurement gates, we minimize the following cost function: 
\begin{align}
    C_{\text{pre}} &= 1 - \left\lvert\sum_{m=1}^{2^r} p(m)\, o_m   \right\rvert.
\end{align}
We show in Appendix \ref{env} that this cost function approximately upper bounds the infidelity. The advantage of minimizing $C_{\text{pre}}$ instead of the infidelity is that the former can be minimized via a generalization of the environment tensor method \cite{evenbly_vidal}, which is more robust than gradient descent, converges faster, and is less sensitive to the barren plateau problem. 
\\ \\ 
Observe that $C_{\text{pre}}$ can be expressed in terms of a single pre-measurement gate, $U_j$, and its environment tensor, $\mathcal{E}_j$:
\begin{align}
    o_m &= \frac{\text{Tr}\left(E_{j;m}\, U_j\right)}{\sqrt{p(m)}}\\
    \Rightarrow C_{\text{env}} &= 1 - \left\lvert\text{Tr}\left(\mathcal{E}_j\, U_j\right) \right\rvert  \\
    \mathcal{E}_j &= \sum_m \sqrt{p(m)}\, E_{j;m} \label{env_eq},
\end{align}
where $E_{j;m}$ can be computed by thinking of $o_m$ as a tensor network, and contracting every tensor except $U_j$ (see Figure \ref{environment}(a)). 
\\ \\ 
Since the dynamic circuits we optimize over are shallow and the target state is assumed to have an efficient matrix product state (MPS) representation, both $E_{j;m}$ and $p(m)$ can be computed quickly on a classical computer and used to construct the environment $\mathcal{E}_j$. We then update $U_j$ using the singular value decomposition of the environment: 
\begin{align}
    \mathcal{E}_j &= XSY^{\dagger} \\ 
    U_j^{\text{new}} &= YX^{\dagger}
\end{align} 
Note that since $U_j$ is applied prior to measurement, it affects the outcome probabilities, $p(m)$ and, therefore, $\mathcal{E}_j$ is not strictly independent of $U_j$. However, our numerical experiments suggest that assuming $p(m)$ remains unchanged during a single gate update does not affect convergence of the algorithm. Under this assumption, the cost function is linear in $U_j$ and, therefore, the singular value decomposition provides the globally optimal update for the gate. 
\\
\\
To optimize the post-measurement gates for outcome $m$, we use a different cost function:
\begin{align}
    C_m &= 1 - \lvert o_m \rvert = 1 - \left\lvert\text{Tr}\left(F_{k;m}\,V_{k;m} \right) \right\rvert,
\end{align}
where $F_{k;m}$ is the standard environment tensor (see Figure \ref{environment}(b)). The optimal update for gate $V_{k;m}$ is also computed by singular value decomposition of its environment. We emphasize that for a look-up table decoder, the optimization of post-measurement gates for each measurement outcome is independent.  
\\ \\ 
A crucial hyperparameter of the DQC ansatz is the placement of the ancilla(e). In our optimization runs, we search over this hyperparameter in a greedy fashion. We find the placement of the first ancilla that results in the highest fidelity. Then we fix that ancilla and consider adding a second ancilla at a different location. This greedy search over hyperparameters incurs an $\mathcal{O}(n)$ overhead for the algorithm without explicit parallelization. 

\subsection{Results}

We apply our algorithm to three classes of states for which we believe explicit DQC constructions have not been studied in the literature.
\\
\begin{figure*}[!t]
  \centering
  \includegraphics[width=\linewidth]{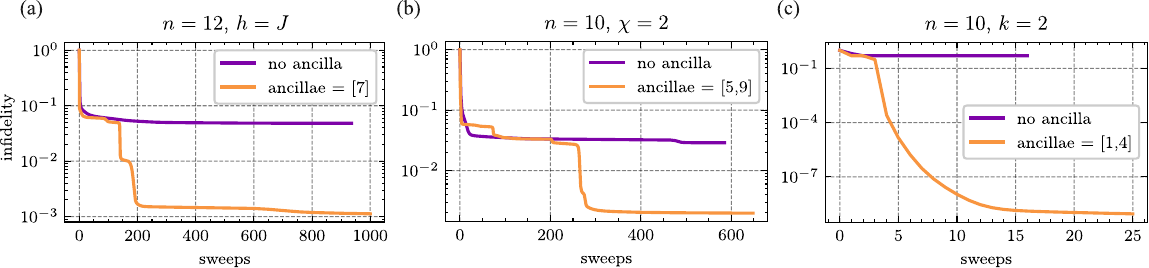}
  \caption{Optimizing static and dynamic quantum circuits of $n$ qubits (not including the ancillae) and depth 4 to prepare (a) ground state of the critical transverse field Ising model, (b) a generic matrix product state, and (c) a hard instance of subset states. In each case, the infidelity of state preparation is reduced by order(s) of magnitude through the introduction of a few ancillary qubits used for measurement and feedforward. Ancillae locations are specified in the legend. Each line represents the best result from 10 independent optimization runs.}
  \label{order_of_magnitude}
\end{figure*}
\\
The first class of states we study consists of ground states of the transverse field Ising (TFI) model at criticality,
\begin{align}
    H = -J\left(\sum_{\langle i,j\rangle} Z_iZ_j + g\sum_j X_j\right),
\end{align}
where $Z_j$ and $X_j$ are the Pauli $Z$ and $X$ operators respectively on the $j$th qubit. At $g=1$ the model undergoes a phase transition and the ground state has long-range correlations. As shown in Appendix \ref{scaling_section}, the entanglement entropy of the state grows as $\mathcal{O}(\log(n))$ while the complexity grows as $\mathcal{O}(n)$. Due to the gap between these two scaling relations, we can expect DQCs to have an advantage over static circuits of the same depth. 
\\ \\ 
The second class of states we consider are random matrix product states (MPS) with fixed bond dimension, $\chi$. These states have an $n$ independent entanglement entropy of $\mathcal{O}(\log{\chi})$ and complexity \cite{log_depth_mps} of $\mathcal{O}(\log{n})$. While DQC preparation of MPSs has been studied in the literature, it has been in the presence of symmetries \cite{mps_from_locc, mps_peps_from_locc}. We consider completely generic MPSs. 
\\
\\ 
Finally we study the recently introduced subset states, defined as an equal superposition of $k$ randomly selected computational basis states: 
\begin{align}
    \ket{\psi} &= \sum_{i=1}^k \frac{1}{\sqrt{k}} \ket{i}.
\end{align}
Subset states were introduced in the context of pseudoentanglement \cite{pseudoentanglement}. Their entanglement entropy is bounded by $\log{k}$ but for large enough $k$ the distribution of subset states well approximates the Haar distribution on quantum states. Note that while some subset states have low circuit complexity, others exhibit long-range entanglement and have complexity given by $\mathcal{O}(n)$. The GHZ state is an example of a $k=2$ subset state with linear complexity. We show in Appendix \ref{scaling_section} that as $n$ increases such hard instances of subset states become an exponentially large fraction of all subset states. 
\\ \\ 
Figure \ref{order_of_magnitude} shows the results for training DQCs for each class of states. We compare DQCs with static circuits of equivalent depth and find that in each case, one or two ancillae are sufficient to lower the infidelity by order(s) of magnitude. The improvement is particularly stark for the subset state we attempt to prepare. We attribute this to the greater gap between the scaling of entanglement entropy and circuit complexity for hard instances of these states. 
\\
\begin{figure}
  \centering
  \includegraphics[width=\linewidth]{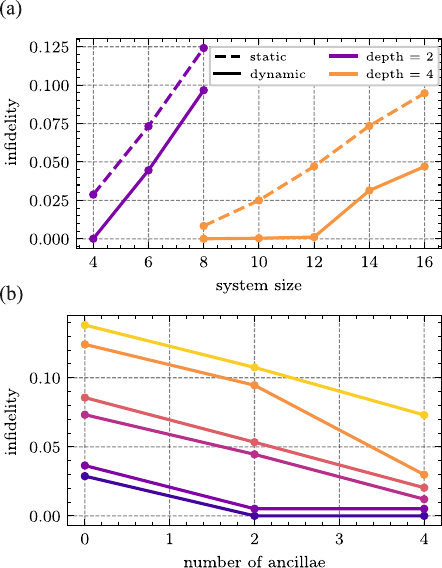}
  \caption{Infidelity of preparing the critical TFI state (a) as a function of system size at different depths, and (b) as a function of number of ancillae for 4 through 9 qubits. Note that the infidelity is largely independent of system size and drops monotonically with the number of ancillae.}
  \label{effect_of_ancillae}
\end{figure}
\\
Figure \ref{effect_of_ancillae}(a) shows that the DQC advantage is robust and consistent across system sizes and ansatz depths. This implies that, regardless of system size, utilizing a few ancilla qubits for measurement and feedforward can enhance the fidelity of state preparation. The data points for dynamic circuits are obtained by optimizing ansatze with a single ancilla to prepare the critical TFI state. Figure \ref{effect_of_ancillae}(b) demonstrates that while adding more ancillae steadily reduces infidelity, a threshold exists beyond which additional ancillae have little effect. 
\\ 
\begin{figure}
  \centering
  \includegraphics[width=\linewidth]{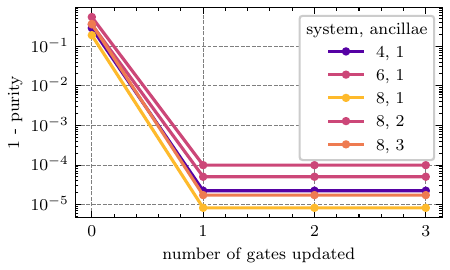}
  \caption{For a range of system sizes and numbers of ancillae, the purity of a randomly initialized dynamic quantum circuit saturates to a value close to 1 after a single gate update.}
  \label{purities}
\end{figure}
\\ 
As noted earlier, the output of a randomly initialized DQC is a mixed state. An important question is how the purity of this state evolves during training. Numerical experiments indicate that the purity rapidly saturates to a value close to 1. Figure \ref{purities} shows that after performing just a single globally optimal update of a pre-measurement gate, the purities have already reached their maximum values. This pattern holds across different system sizes and numbers of ancillae.  

\section{DQC with many ancillae \label{extensive_ancillae}}

The literature on state preparation with dynamic quantum circuits has so far focused on DQCs that offer constant depth preparations of states. Such constructions promise impressive improvements in terms of fidelity. But, in all known instances, the number of ancillae required scales linearly with system size. This makes a look-up table decoder infeasible since one must store and train $\mathcal{O}(2^n)$ number of parameters. 
\\ \\ 
The challenge of decoding exponentially many syndromes has already been studied in the context of quantum error correction \cite{decoder_review}. In this section we propose two methods to decode DQCs with an extensive number of ancillae, by drawing inspiration from decoders used for error correction. 

\subsection{Neural network decoder}

The classical decoder for a dynamic circuit is a function from $r$ integers to $\mathcal{O}(\text{poly}(n))$ real numbers. The universal approximation theorem \cite{nn_universal} tells us that with enough parameters a neural network can approximate this function. With this in mind, we propose a feedforward neural network decoder \cite{nn_decoder} with an input node for each ancilla and output nodes for parameters of the adaptive circuit. We assume that a poly$(n)$ number of hidden neurons suffice to approximate the decoding function. 
\\ 
\\
We perform stochastic gradient descent on the following loss function to train the neural network: 
\begin{align}
    L &= \sum_{m\in M} \left\lvert \bra{\psi_{\text{targ}}} \prod_k f(m)_k\,\frac{\bra{m}}{\sqrt{p(m)}}\,\prod_j U_j\ket{0}^{\otimes (n+r)} \right\rvert^2, 
\end{align}
where $f(m)_k$ represents the $k$th gate returned by the neural network given outcome $m$. For each epoch of training, instead of summing over all measurement outcomes to compute the loss, we stochastically generate a batch of outcomes, $M$, sampled by measuring $\prod_j U_j \ket{0}^{n+r}$. Our results indicate that training succeeds even when $|M| \ll 2^r$. The gradients with respect to weights and biases of the neural network are computed via automatic differentiation. 
\\ 
\begin{figure}
  \centering
  \includegraphics[width=\linewidth]{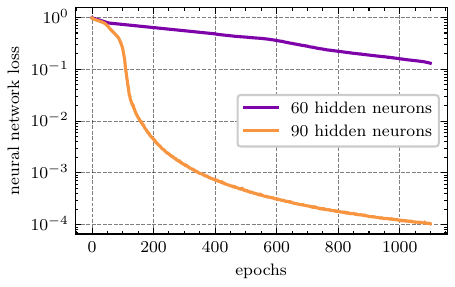}
  \caption{Training a feedforward neural network decoder to prepare a 6 qubit GHZ state. Increasing the number of hidden neurons allows the neural network to converge faster to an approximation of the decoding function.}
  \label{nn}
\end{figure}
\\
In Figure \ref{nn} we show that a feedforward neural network can be trained to decode measurement outcomes for a DQC preparing the GHZ state. The pre-measurement circuit is taken from \cite{aklt_from_locc} and constructs small GHZ states on local three qubit patches. The boundaries of neighboring patches are then measured in the Bell basis and the outcome passed to the neural network. The plots shows that increasing the number of hidden neurons improves convergence of the decoder. 
\\ 
\begin{figure*}[!t]
  \centering
  \includegraphics[width=0.75\linewidth]{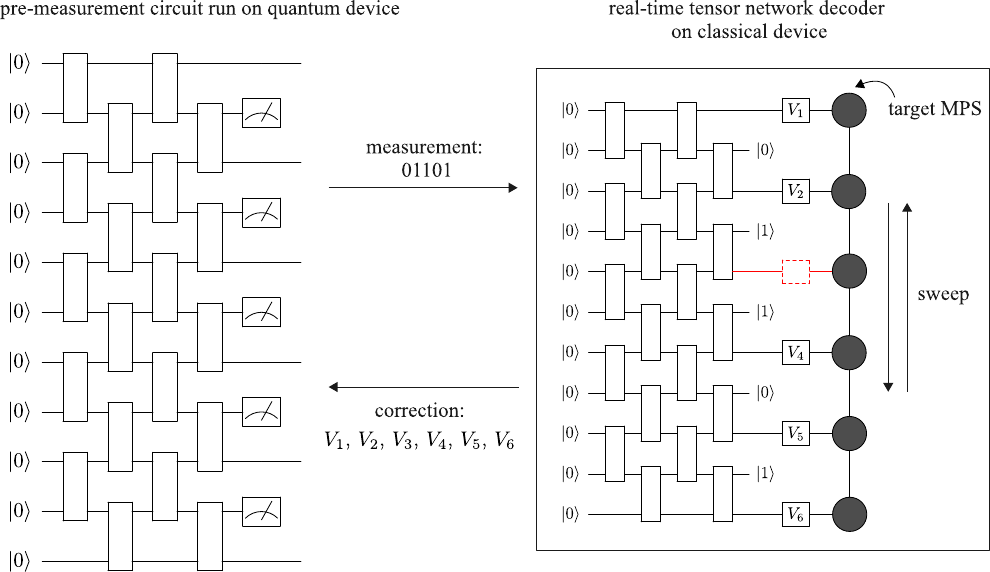}
  \caption{Real time decoding of a dynamic circuit for state preparation. The pre-measurement circuit is executed on a quantum computer and the ancillae are measured. The outcome is sent to a classical device that simulates the post-measurement state, the adaptive circuit and the target state using tensor networks. Each gate in the adaptive circuit is optimized by computing its environment tensor and performing singular value decomposition.}
  \label{real_time_decoding}
\end{figure*}
\\ 
An important future direction is to learn suitable pre-measurement circuits for the neural network decoder given arbitrary target states. Although this concept is not explored in the current work, we believe a stochastic version of the environment tensor method described in Section \ref{few_ancillae} could be effective. Specifically, the sum in Equation \eqref{env_eq} can be computed over $\text{poly}(n)$ sampled measurement outcomes rather than all $2^r$ outcomes, allowing the environment tensor-based updates to be performed in polynomial time.   

\subsection{Real-time decoding \label{real_time_section}}

An alternative method for dealing with an exponentially large number of measurement outcomes is to relax the condition that the decoder must succeed for all possible outcomes. In fact it is sufficient to find the adaptive circuit only for the measurement outcome that has occurred on the device. Commonly used decoders in quantum error correction, like minimum weight perfect matching (MWPM) \cite{mwpm}, similarly decode only the observed syndrome
in real-time on a classical device next to the quantum computer. 
\\ 
\begin{figure}
  \centering
  \includegraphics[width=\linewidth]{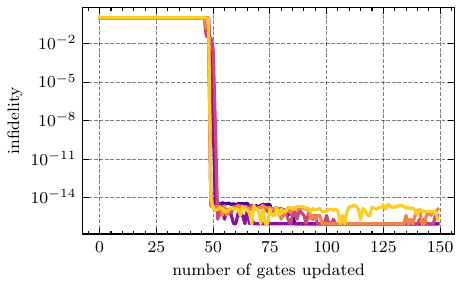}
  \caption{Real time decoding of a 50 qubit GHZ state for five randomly selected measurement outcomes. In each case, the adaptive circuit is found within a single sweep of the gates.}
  \label{real_time}
\end{figure}
\\ 
Figure \ref{real_time_decoding} describes a real-time decoding protocol for state preparation with dynamic circuits. The pre-measurement circuit is run on a quantum device and the ancillae are measured with the outcome passed to a classical device. Both the post-measurement state and the target state are then reconstructed as matrix product states, connected by parametrized gates. We then sweep over the gates and update each using the singular value decomposition of its environment tensor. Once optimized, the resulting adaptive circuit is passed to the quantum device and applied to the unmeasured qubits. 
\\ \\ 
We benchmark this protocol by simulating the decoding of a 50 qubit GHZ state. The pre-measurement circuit is again taken from \cite{aklt_from_locc} and the adaptive circuit optimized over consists of single qubit rotations. Figure \ref{real_time} shows that for several randomly chosen measurement outcomes, the correct adaptive circuit is found within a single sweep of the gates. This speed of convergence is essential for the success of real-time decoding since the unmeasured qubits decohere while the decoding takes place.    
\\ \\ 
The decoding protocol we describe converges quickly under the assumption that the MPS representations of the target state and post-measurement states have low bond dimension, and that the adaptive circuit is shallow. In Appendix \ref{stabilizers} we show that these requirements are satisfied when the target is a stabilizer state with low entanglement. Specifically, we prove that measuring a subset of the qubits in a stabilizer state results in post-measurement states that can be mapped onto each other by single qubit Pauli gates. 
\\ \\ 
Therefore, to find a DQC construction of a stabilizer state with low entanglement, it is sufficient to optimize over a Clifford pre-measurement circuit \cite{clifford_optimize1, clifford_optimize2} such that for a single fixed measurement outcome, it prepares the target state with high fidelity. Theorem \ref{stabilizer_thm} then guarantees that for any other measurement outcome, the real-time decoding protocol described above can rapidly find an adaptive circuit that maps the post-measurement state to the target state. 
\\ \\ 
This approach has immediate practical applications since stabilizer states with low entanglement encompass many physically relevant quantum states. This includes the GHZ state as well as ground states of stabilizer Hamiltonians such as those describing fractons \cite{toric_code, fractons}. However, extending this framework beyond stabilizer states remains an interesting future direction.     

\section{Discussion and outlook \label{outlook}} 

Our work demonstrates that it is possible to use classical algorithms to learn dynamic quantum circuits for quantum state preparation. We find that by utilizing a few ancilla qubits for measurement and feedforward, one can achieve significant improvements in state preparation fidelity. This means that in the presence of low readout error and effective dynamical decoupling \cite{cycle, pauli_learning, randomized, readout}, there is no real drawback and much to gain by using DQCs to prepare states.
\\ \\ 
In order to make our algorithms scalable, we extended the environment tensor optimization technique to a setting that includes measurement and feedforward. This allowed us to avoid the barren plateau problem that plagues gradient based methods. We have also developed methods to overcome the exponential growth of measurement outcomes in DQCs with an extensive number of ancillae. 
\\ \\ 
An immediate extension of this work would be to apply our algorithms to higher dimensions, which offer a richer landscape of long-range entangled states. A generalization of the ideas in this paper to unitary and channel compilation would be natural as well. There is also room for exploring applications of DQCs to machine learning problems. 
\\ \\ 
While our focus in this work has been on classical algorithms for training DQCs, studying variational quantum algorithms -- where part of the optimization takes place on a quantum computer -- might prove fruitful. In fact quantum optimization might become necessary for higher dimensional problems, where tensor network techniques are far less efficient. 
In a recent work \cite{vqe_locc}, the authors have built some of the necessary tools to implement such algorithms. 
\\ \\ 
From a theoretical perspective, understanding the limitations of the dynamic circuit ansatz remains crucial. An interesting open question is understanding what states can be prepared to high fidelity by dynamic quantum circuits of a given depth. Our algorithmic approach may facilitate numerical explorations of this question, leading to a theory of dynamical quantum circuit complexity.  
\\ \\ 
As dynamic quantum circuits are necessary en route to fault-tolerant quantum computation, it is natural to explore what other applications are possible with this new technology. We expect our algorithmic approach to be a key enabling technique going forward in this quest. 

\section*{Acknowledgements}
F.A. was supported by Laboratory Directed Research and Development program of Los Alamos National Laboratory under project number 20230049DR. B.K.C acknowledges support from the NSF Quantum Leap Challenge Institute for Hybrid Quantum Architectures and Networks (NSF Award No. 2016136). Tensor network computations were done with the ITensors.jl library \cite{itensor}. F.A. thanks Shivan Mittal, Felix Leditzky and Lukasz Cincio for helpful discussions. B.K.C and F.A. also thank Abid Khan for helpful discussions. 

\bibliography{lib.bib}

\appendix

\section{Cost function \label{env}}

In Section \ref{few_ancillae}, we use the following cost function to optimize pre-measurement gates: 
\begin{align}
    C_{\text{pre}} &= 1 - \left\lvert\sum_{m=1}^{2^r} p(m)\, o_m   \right\rvert,
\end{align}
where $p(m)$ is the probability of outcome $m$ and $o_m$ is the overlap of the corresponding post-measurement state with the target state. We can use the triangle inequality to get the following bound:
\begin{align}
    \left\lvert\sum_{m=1}^{2^r} p(m)\, o_m   \right\rvert &\le \sum_{m=1}^{2^r}p_m\, \lvert o_m \rvert \\ 
    \Rightarrow\qquad C_{\text{pre}} &\ge 1 - \sum_{m=1}^{2^r}p_m\, \lvert o_m \rvert
\end{align}
The expression on the right is equal to the square root infidelity: $1-\sqrt{\bra{\psi_{\text{targ}}}\rho_{\text{prep}}\ket{\psi_{\text{targ}}}}$. This quantity serves as a measure of dissimilarity between quantum states, and our cost function upper bounds it. However, it is only the true infidelity, $1-\bra{\psi_{\text{targ}}}\rho_{\text{prep}}\ket{\psi_{\text{targ}}}$, that has an interpretation as a statistical distance \cite{nielsen_chuang}. 
\\ \\ 
To relate $C_{\text{pre}}$ to infidelity, we note that when $\lvert o_m \rvert$ approaches 1, it holds that $\lvert o_m \rvert \approx \lvert o_m \rvert^2$. Thus, in the low infidelity regime, infidelity is well approximated by the square root infidelity, allowing us to approximately upper bound the former in terms of $C_{\text{pre}}$: 
\begin{align}
    C_{\text{pre}} \gtrsim 1 - \sum_{m=1}^{2^r}p_m\, \lvert o_m \rvert^2 = \text{infidelity}.
\end{align}

\section{Scaling of entanglement and complexity \label{scaling_section}} 

Enhancing a quantum circuit of some depth with measurement and feedforward does not increase its entangling power. This is because measurement and feedforward is an example of a local operation with classical communication (LOCC) under which bipartite entanglement is non-increasing \cite{locc_basics}. As a result, for highly entangled target states, dynamic circuits provide no advantage over their static counterparts. 
\\ \\ 
Instead LOCC allows DQCs to introduce long-range correlations which may be disallowed for static circuits of the same depth due to the Lieb-Robinson bound \cite{lieb_robinson}. Therefore, a method to determine if one should use DQCs to prepare a class of states is to compare the scaling of its \textit{entanglement depth} and its \textit{preparation depth} or complexity. Here by entanglement depth we mean the minimum depth circuit necessary to produce a state with as much bipartite entanglement as the target state across any bipartition. We define complexity to be the minimum depth circuit necessary to prepare the state to a chosen fidelity threshold.  
\\ 
\begin{figure*}[!t]
  \centering
  \includegraphics[width=\linewidth]{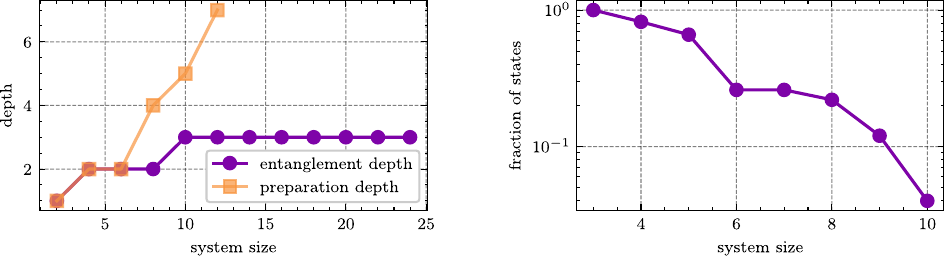}
  \caption{Scaling of entanglement and complexity in target states. (a) The preparation depth of a critical TFI state grows linearly with system size, while the entanglement depth only grows logarithmically. The preparation depth is computed by optimizing a static circuit to a fidelity threshold of $0.999$. (b) The fraction of 50 randomly generated $k=2$ subset state samples that can be prepared with a depth 2 static circuit decays exponentially with system size.}
  \label{scaling}
\end{figure*}
\\
In Figure \ref{scaling}(a), we present numerical experiments showing that the entanglement depth of the critical transverse field Ising model grows slower than its preparation depth, due to the long-range correlations present at criticality. The existence of this gap allows room for a DQC advantage. 
\\ \\ 
The equivalent analysis for subset states is made difficult by the large variance in the complexity across instances. While the entanglement depth is given by $\log{k}$ for states with subset size $k$, complexity can vary from $\mathcal{O}(1)$ for simple subsets to $\mathcal{O}(n)$ for GHZ-like states. 
\\ \\
However, our numerical experiments demonstrate that instances with greater preparation depth than entanglement depth quickly become the norm as system size increases. Figure \ref{scaling}(b) shows that if we fix the subset size to $k=2$, the fraction of sampled subset states that can be prepared by a depth 2 static circuit decays exponentially with system size. This suggests that for large system sizes, a random subset state will very likely exhibit long-range correlation. DQCs can improve the fidelity of preparing such states. 

\section{Decoding Clifford DQCs \label{stabilizers}}

The real-time decoding protocol described in Section \ref{real_time_section} works under the assumption that one knows or can find a pre-measurement circuit such that every post-measurement state can be related to the target state via a simple circuit. Here we show that this is always possible for stabilizer target states and stabilizer dynamic quantum circuits.
\\ \\ 
Consider an $n$-qubit stabilizer state $\ket{\psi}$ with stabilizers: $S=\langle g_1,g_2,\dots, g_n \rangle$, where $g_i$ are independent elements of the Pauli group on $n$-qubits that mutually commute. Suppose we measure the $k$th qubit in the computational basis and consider the two post-measurement states: $\ket{\phi_0}$ and $\ket{\phi_1}$. We want to show that for every $\ket{\psi}$, there exists an element $g$ of the Pauli group $\mathcal{P}_n$ such that $\ket{\phi_0} = g\ket{\phi_1}$. 
\\ \\ 
First observe that $\ket{\phi_0}$ and $\ket{\phi_1}$ are stabilizer states that share all stabilizer generators except one \cite{nielsen_chuang}: $\langle Z_k, g_2, \dots, g_n \rangle$ and $\langle -Z_k, g_2, \dots, g_n \rangle$. We will, therefore, show that a Pauli string is sufficient to map one set of generators into the other. For this we will need to use the check matrix representation of stabilizer states. 
\\ \\ 
An $n$-qubit Pauli string $g\in \mathcal{P}_n$ has the binary representation $r(g) = (r_x(g), r_z(g))$. $r_x(g)$ is an $n$-bit binary whose $i$th bit is 1 if $g$ applies an $X$ or $Y$ gate on the $i$th qubit and zero otherwise. $r_z(g)$ is also an $n$-bit binary whose $i$th bit is 1 if $g$ applies a $Z$ or $Y$ gate on the $i$th qubit.
\\ \\ 
Two Pauli strings, $g,g'$ commute if and only if $r(g)\Lambda r(g')=0$ where $\Lambda$ is a $2n\times 2n$ matrix that \textit{swaps} the $X$ and $Z$ degrees of freedom: 
    \begin{align*}
        \Lambda &= \begin{pmatrix}
            0 & I_n \\ 
            I_n & 0
        \end{pmatrix},
    \end{align*}
and the inner product is computed modulo 2. $r(g)\Lambda r(g')$ essentially counts the number of qubits at which $g,g'$ share distinct, non-identity Pauli gates.  
\\ \\ 
The check matrix of a stabilizer state is an $n\times 2n$ matrix whose rows contain the binary representation of each stabilizer generator of $\ket{\psi}$: 
\begin{align}
    R_{\psi} &= \begin{pmatrix}
        r(g_1) \\ 
        r(g_2) \\
        \vdots \\ 
        r(g_n)
    \end{pmatrix}
\end{align}

With this setting, the proof of our statement follows from these lemma:

\begin{lemma}\label{indep}
If $g_1,g_2,\dots,g_n$ are independent, the rows of $R_{\psi}$ are linearly independent. 
\end{lemma}

\begin{proof}
Note that in the binary representation, product of Pauli strings turn into addition: $r(g_1g_2) = r(g_1) + r(g_2)$, where the addition is modulo 2. Also, exponentiation turns into scaling: $r(g^a)=ar(g)$. 
\\ \\ 
Since $g_1,g_2,\dots,g_n$ are independent, the only way to satisfy $\prod_i g_i^{a_i} = I$ is to have $a_i=0$ for every $i$. Taking the binary representation of both sides, we get: 
\begin{align*}
    r\left(\prod_ig_i^{a_i}\right) &= r(I) \\ 
    \sum_i a_i r(g_i) &= 0. 
\end{align*}
The condition on $a_i$ then translates to $r(g_1),r(g_2),\dots,r(g_n)$ forming a linearly independent set of vectors.
\\ \\ 
Note that since $-I$ (and by extension $\pm iI$) are never stabilizers of a non-zero state, this is also an equivalence statement since the only $g$ with $r(g)=0$ is $I$, which means the second line above implies the first.  
\end{proof} 

\begin{lemma}\label{anticommute}
    If $g_1,g_2,\dots,g_n$ are independent, for every $i$ in $1,\dots,n$, there exists $g\in \mathcal{P}_n$ that anti-commutes with $g_i$ but commutes with every other generator. 
\end{lemma}

\begin{proof}
    Given $g_i$ we provide a construction for the $g$ that satisfies this property. 
    \\ \\ 
    Suppose we find an $r$ that solves the equation: $R_{\psi}\Lambda r = e_i$, where $e_i$ is a vector of length $n$ that has 1 as the $i$th entry and 0 everywhere else. $r$ then represents the binary representation of a Pauli string that anti-commutes only with the $i$th generator of $\ket{\psi}$ but commutes with all the rest.  
    \\ \\ 
    Since the generators are assumed to be independent, by Lemma \ref{indep}, $R_{\psi}$ has $n$ linearly independent rows. Therefore, it has a right inverse, $R_{\psi}^{-1}$. $\Lambda$ is a symmetric matrix with all eigenvalues equal to 1 and is, therefore, invertible. We see then that $r=\Lambda^{-1}R_{\psi}^{-1}e_i$ provides at least one solution to the equation we want to solve.  
\end{proof}

\begin{theorem}\label{stabilizer_thm}
For every $n$-qubit stabilizer state $\ket{\psi}$ and $k\in[1,n]$, there exists $g\in\mathcal{P}_n$ such that $\ket{\phi_0}=g\ket{\phi_1}$, where $\ket{\phi_0}$ and $\ket{\phi_1}$ are the two post-measurement states resulting from measuring the $k$-th qubit of $\ket{\psi}$.  
\end{theorem}  

\begin{proof}
We use the method in Lemma \ref{anticommute} to find a $g$ that anti-commutes with $Z_k$ but commutes with all other stabilizers. This $g$ then maps $\ket{\phi_0}$ to $\ket{\phi_1}$.  
\end{proof}

Finding a class of parametrized, non-Clifford circuits for which one can prove a similar result would make our real-time decoding protocol applicable to non-stabilizer states.  

\end{document}